\documentclass[letterpaper, 10 pt, conference]{ieeeconf}  

\usepackage{amssymb, latexsym, mathrsfs}
\usepackage{amsmath}
\usepackage{mathtools}
\usepackage{dsfont}
\usepackage{enumerate}
\usepackage{algorithm}
\usepackage{color}
\usepackage{transparent}
\usepackage{multicol}
\usepackage{booktabs}
\usepackage{etex}
\usepackage{shellesc}
\usepackage{subcaption}

\usepackage{pgfplots}
\usepackage{tikz}
\newlength\fheight 
\newlength\fwidth 
\usepackage{indentfirst}
\pgfplotsset{compat=newest} 
\pgfplotsset{plot coordinates/math parser=false}
\usepackage{tikzscale}
\usepackage[american resistor, european voltage, european current]{circuitikz}
\usetikzlibrary{arrows,shapes,positioning}
\usetikzlibrary{decorations.markings,decorations.pathmorphing,
	decorations.pathreplacing}
\usetikzlibrary{calc,patterns,shapes.geometric}
\usetikzlibrary{arrows}
\usetikzlibrary{shapes.misc}
\usetikzlibrary{arrows.meta,positioning}

\newtheorem{thm}{Theorem}
\newtheorem{assum}{Assumption}

\newtheorem{rmk}{Remark}

\newtheorem{lem}{Lemma}
\newtheorem{defn}{Definition}

\makeatletter
\let\NAT@parse\undefined
\makeatother
\usepackage[bookmarks=false]{hyperref}

\usepackage{indentfirst}
\pgfplotsset{plot coordinates/math parser=false}
\usetikzlibrary{positioning}
\IEEEoverridecommandlockouts                              

\newcommand{\Rset}{\mathbb{R}}

\newcommand{\BB}{{\mathcal{B}}}
\newcommand{\CC}{{\mathcal{C}}}
\newcommand{\DD}{{\mathcal{D}}}

\newcommand{\GG}{{\mathcal{G}}}
\newcommand{\HH}{{\mathcal{H}}}
\newcommand{\II}{{\mathcal{I}}}
\newcommand{\JJ}{{\mathcal{J}}}

\newcommand{\LL}{{\mathcal{L}}}

\newcommand{\NN}{{\mathcal{N}}}

\newcommand{\RR}{{\mathcal{R}}}

\newcommand{\ZZ}{{\mathcal{Z}}}

\newcommand{\mbf}[1]{\mathbf{#1}}                  

\newcommand{\subss}[2]{{#1}_{[#2]}}

\title{\LARGE \bf
On Existence of Equilibria, Voltage Balancing, and Current Sharing in Consensus-Based DC Microgrids*
}

\author{Pulkit Nahata and Giancarlo Ferrari-Trecate$^{1}$
	\thanks{*This work has received support from the Swiss National Science Foundation under the COFLEX project (grant number 200021{\_}169906).}
	\thanks{$^{1}$Pulkit Nahata, and Giancarlo Ferrari-Trecate are with the Dependable Control and Decision group (DECODE) of the  Automatic Control Laboratory, \'Ecole Polytechnique F\'ed\'erale de Lausanne (EPFL), Switzerland. Email addresses: {\tt\small \{pulkit.nahata, giancarlo.ferraritrecate\}@epfl.ch}}%
}

\begin{document}

\maketitle
\thispagestyle{empty}
\pagestyle{empty}

\begin{abstract}

This work presents new secondary regulators for current sharing and voltage balancing in DC microgrids, composed of distributed generation units, dynamic RLC lines, and nonlinear ZIP (constant impedance, constant current, and constant power) loads. The proposed controllers sit atop a primary voltage control layer, and utilize information exchanged over a communication network to take necessary control actions. We deduce sufficient conditions for the existence and uniqueness of an equilibrium point, and show that the desired objectives are attained in steady state.  Our control design only requires the knowledge of local parameters of the generation units, facilitating plug-and-play operations. We provide a voltage stability analysis, and illustrate the performance and robustness of our designs via simulations. All results hold for arbitrary, albeit connected, microgrid and communication network topologies.

\end{abstract}

\section{INTRODUCTION}
\label{sec:introduction}
Microgrids (mGs) are electric networks comprising distributed generation units (DGUs), storage devices, and loads. Apart from their manifold advantages like integration of renewables, enhanced power quality, reduced transmission losses, capability to operate in grid-connected and islanded modes, they are compatible with both AC and  DC operating standards \cite{Bhaskara}. In particular, DC microgrids (DCmGs), have gained traction in recent times. Their rising popularity can be attributed to development of efficient converters, natural interface with renewable energy sources (for instance PV modules) and batteries, and availability of electronic loads (various appliances, LEDs, electric vehicles, computers etc) inherently DC in nature \cite{Dragicevic1, Meng}. 

In islanded DCmGs, voltage stability is crucial, for without it voltages may breach a critical level and damage connected loads \cite{Meng}. Thus, a primary voltage control layer is often employed to track desired voltage references at points of coupling (PCs), whereby DGUs are connected to the DCmG. To this aim, several approaches, for example, based on droop control \cite{Dragicevic1,Shafiee2014} and plug-and-play control \cite{Nahata, Tucci2016independent}, have been proposed in the literature. Besides voltage stability, another desirable objective is to ensure current sharing, that is, DGUs must share mG loads in accordance with their current ratings. Indeed, in its absence, unregulated currents may overload generators and eventually lead to an mG failure. An additional goal of voltage balancing, requiring boundedness of weighted sum of PC voltages, is often sought to complement current sharing \cite{Tucci2018}. Primary controllers, however, are blind voltage reference emulators and are unable to attain the aforementioned objectives by themselves. Higher-level secondary control architectures are therefore necessary to coordinate the voltage references provided to the primary layer. 

Consensus-based secondary regulators guaranteeing current sharing and voltage balancing have been the subject of many recent contributions. Centralized design approaches are proposed in \cite{Nasirian, shafiee2014distributed_b}, but are prohibitive for large-scale mGs as they require particulars of mG topology, lines, loads, and DGUs. With the aim of overcoming this issue, scalable design procedures \cite{Nahata, Tucci2016independent} have been proposed, which enable the synthesis of decentralized controllers and plug -in/-out of DGUs on the fly without spoiling the overall stability of the network. Distributed consensus-based controllers for DCmGs with generic topologies, discussed in \cite{Tucci2018, Zhao}, remedy the limitations of centralized design schemes, but presume static lines and abstract DGUs as ideal voltage generators or first-order systems. Efforts to take into account DGU dynamics and RL lines have been made in \cite{Trip, Cucuzzella2018robust}. In \cite{Cucuzzella2018robust}, a robust distributed control algorithm is proposed considering both objectives; however, a suitable initialization of the controller is needed. The resistance of the DGU filter is neglected in \cite{Trip} and hence, voltage balancing cannot be guaranteed in steady state. In addition, the above-mentioned contributions \cite{Tucci2018,Zhao,Trip,Cucuzzella2018robust} are limited to linear loads. A power consensus algorithm with ZIP loads is studied in \cite{DePersis} under simplified DCmG dynamics and assumptions on existence of a suitable steady state.

\subsection{Paper Contributions}

In this paper, we build upon our previous theoretical contributions on primary voltage control \cite{Nahata}, and introduce a distributed secondary control layer for proportional current sharing and weighted voltage balancing in DCmGs consisting of DGUs, loads, and interconnecting power lines. 

The main technical novelties of this paper are fourfold. First, this work does away with the modeling limitations of several existing contributions. In addition to dynamic RLC lines and nonlinear ZIP loads, we consider DGUs interfaced with DC-DC Buck converters  and incorporate complete converter dynamics along with filter resistances. Second, we propose a new consensus-based secondary control scheme relying on the exchange of variables (obtained from DGU filter currents) with nearest communication neighbors. These secondary regulators appropriately modify primary voltage references to attain the desired goals. In spite of their distributed structure, the control design is completely decentralized allowing for plug-and-play operations. Third, we throughly investigate the steady-state behavior of the DCmG under secondary control, and deduce sufficient conditions on the existence and uniqueness of an equilibrium point meeting secondary goals. {Such an analysis is not trivial due to the nonlinearities introduced by the ZIP loads and entails finding solutions of the DC power-flow equations under consensus constraints on a hyperplane.} To the best of our knowledge, this has not been addressed in the literature \cite{Bolognani, simpson2016voltage} before.  Finally, we present conditions on the controller gains and power consumption of P loads guaranteeing voltage stability of the closed-loop DCmG, and show that stability is independent of DCmG and communication topologies.

The remainder of Section \ref{sec:introduction} introduces relevant preliminaries and notation. Section \ref{sec:model} recaps the DCmG model and primary voltage control. Section \ref{sec:secondary} houses our main contributions. It presents consensus-based secondary controllers, and details the steady-state behavior and stability of the closed-loop DCmG in the presence of ZIP loads. Simulations validating theoretical results are provided in Section \ref{sec:simulations}.  Finally, conclusions are drawn in Section \ref{sec:conclusions}. 

\subsection{Preliminaries and notation}
\textit{Sets, vectors, and functions:} We let $\mathbb{R}$ (resp. $\mathbb{R}_{>0}$) denote the set of real (resp. strictly positive real) numbers. For a finite set $\mathcal{V}$, let $|\mathcal{V}|$ denote its cardinality. Given $ x \in \mathbb{R}^{n}$, $[x] \in \mathbb{R}^{n \times n}$ is the associated diagonal matrix with $x$ on the diagonal. The inequality $x\leq y$ is component-wise, that is, $x_i\leq y_i,~\forall i\in 1,...,n$.  Throughout, $\textbf{1}_n$ and $\textbf{0}_n$ are the $n$-dimensional vectors of unit and zero entries, and $\mbf{0}$ is a matrix of all zeros of appropriate dimensions. 
The average of a vector $v\in\mathbb{R}^n$ is $\langle v\rangle=\frac{1}{n}\sum_{i=1}^n v_i$. We denote with $H^1$ the subspace composed by all vectors with zero average  i.e. $H^1 = \{v\in\mathbb{R}^n:\langle v\rangle = 0\}$. For the matrix $A\in\mathbb{R}^{m\times n}$, $A^\dagger\in \mathbb{R}^{n \times m}$ denotes its pseudo inverse whereas its range and null spaces are indicated by $\mathcal{R}(A)$ and $\mathcal{N}(A)$, respectively. 

\textit{Algebraic graph theory}: We denote by $\mathcal{G}(\mathcal{V},\mathcal{E},{W})$ an undirected graph, where $\mathcal{V}$ is the node set and $\mathcal{E}=(\mathcal{V}\times\mathcal{V})$ is the edge set. If a number $l \in \{1,...,|\mathcal{E}|\}$ and an arbitrary direction are assigned to each edge, the incidence matrix $B \in \mathbb{R}^{|\mathcal{V}|\times|\mathcal{E}|}$ has non-zero components: $B_{il} = 1$ if node $i$ is the sink node of edge \textit{l}, and $B_{il} =-1$ if node $j$ is the source node of edge $l$. The \textit{Kirchoff's Current Law} (KCL) can be represented as $x = B\xi$, where $x \in \mathbb{R}^{|\mathcal{V}|}$ and $\xi \in \mathbb{R}^{|\mathcal{E}|}$ respectively represent the nodal injections and edge flows. Assume that the edge $l \in \{1,...,|\mathcal{E}|\}$ is oriented from $i$ to $j$, then for any vector $V \in \mathbb{R}^{|\mathcal{V}|}$, $(B^TV)_l=V_i-V_j$. The Laplacian matrix $\mathcal{L}$ of graph $\GG$ is $\LL=BWB^T$. If the graph is connected, then $\RR(\LL)=H^1$.

\section{DCmG model and primary voltage control}
\label{sec:model}
In this section, we  start by reviewing our DCmG model \cite{Nahata, Tucci2016independent} comprising multiple DGUs interconnected with each other via power lines, and recall the concepts of primary voltage control.  

\textit{DCmG Model:} {The DCmG is modeled as an undirected connected graph $\mathcal{G}_{e}=(\mathcal{D},\mathcal{E})$, where $\DD=\{1,\dots,N\}$ is the node set and $\mathcal{E}\subseteq\mathcal{D}\times\mathcal{D}$ the edge set. Each node of the DCmG is connected to a DGU and a load, and forms the $i^{th}$ PC.  The interconnecting power lines are represented by the edges of $\mathcal{G}_e$. On assigning a number to each line, one can equivalently express $\mathcal{E} = \{1,\dots,M\}$ with $M$ denoting the total number of lines. Note that edge directions are assigned arbitrarily, and provide a reference system for positive currents.  We refer the reader to Figure \ref{fig:powernework} for a representative diagram of the DCmG.}      

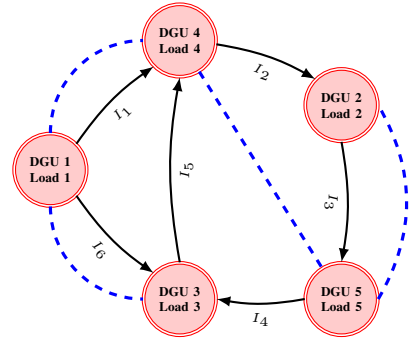
\begin{figure}[h!]
	\centering
	\ctikzset{bipoles/length=1.2cm}
	\tikzstyle{every node}=[font=\tiny]
	\vspace{-0.1cm}
	\begin{circuitikz}[american currents, scale=0.43]
		\tikzstyle{DGU} = [circle, draw, double, align=center, fill=red!20, draw=red]
		
		\draw (-1,1) node(D1) [DGU]  {\textbf{DGU 1} \\ \textbf{Load 1}};
		\draw (8,3) node(D2) [DGU]  {\textbf{DGU 2}\\  \textbf{Load 2}};
		\draw (3,-3) node(D3) [DGU]  {\textbf{DGU 3} \\ \textbf{Load 3}};
		\draw (3,5)  node(D4) [DGU]  {\textbf{DGU 4} \\ \textbf{Load 4}};
		\draw (8,-3)node(D5) [DGU]  {\textbf{DGU 5} \\ \textbf{Load 5}};
		
		\tikzstyle{every node}=[font=\tiny]		
		\path [-latex,thick] (D1.north east) edge [ bend left=10] node[below, sloped] {$I_1$} (D4.south west);	
		\path [-latex,thick] (D4) edge [ bend left=10] node[below, sloped] {$I_2$} (D2.north west);	
		\path [-latex,thick] (D2.south) edge [ bend left=10] node[below, sloped] {$I_3$} (D5.north);	
		\path [-latex,thick] (D5.west) edge [ bend left=10] node[below, sloped] {$I_4$} (D3.east);	
		\path [-latex,thick] (D3.north) edge [ bend left=10] node[below, sloped] {$I_5$} (D4.south);	
		\path [-latex,thick] (D1.south east) edge [ bend right=10] node[below, sloped] {$I_6$} (D3.north west);	
		\path[dashed, very thick, blue] (D1) edge [ bend left=45] (D4);	
		\path[dashed, very thick, blue] (D1) edge [ bend right=45] (D3);	
		\path[dashed, very thick, blue] (D4) edge  (D5);	
		\path[dashed, very thick, blue] (D5.east) [ bend right=30] edge  (D2.east);	
		
	\end{circuitikz}
	\caption{A representative diagram of the DCmG  with the communication network appearing in dashed blue.}
	\label{fig:powernework}
\end{figure}
\textit{Dynamic model of a power line:} The power lines are modeled after the $\pi$-equivalent model of transmission lines \cite{Kundur}. It is assumed that the line capacitances are lumped with the DGU filter capacitance (capacitor $C_{ti}$ in Figure \ref{fig:ctrl_complete}). Therefore, as shown in Figure \ref{fig:ctrl_complete}, the power line $l$ is modeled as a $RL$ circuit with resistance $R_l>0$ and inductance $L_l>0$. By applying Kirchoff's voltage law (KVL) on the $l^{th}$ line, one obtains
\begin{equation}
\begin{small}
\label{eq:powerline}
\subss{{\Sigma}}{l}^{Line}:\left\{\dfrac{dI_l}{dt} = - \dfrac{R_l}{L_l}I_l + \dfrac{1}{L_l}\sum_{i\in\NN_l}B_{il}V_i\right., \\
\end{small}
\end{equation} 
where the variables $V_i$ and $I_l$ represent the voltage at $PC_i$ and the current flowing through the $l^{th}$ line, respectively. 

\textit{Dynamic model of a DGU:} The DGU comprises a DC voltage source (usually generated by a renewable resource),  a Buck converter, and a series $RLC$ filter.  The $i^{th}$  DGU feeds a local load at $PCC_i$ and is connected to other DGUs through power lines.  A schematic electric diagram of the $i^{th}$ DGU along with load, connecting line(s), loads, and local PnP voltage controller is represented in Figure \ref{fig:ctrl_complete}. On applying KCL and KVL on the DGU side at $PC_i$, we obtain
\begin{small}
\begin{equation}
\label{eq:DGUdynamics}
\subss{{\Sigma}}{i}^{DGU}:
\left\{
\begin{aligned}
C_{ti}\dfrac{dV_{i}}{dt} &= I_{ti}-{I_{Li}(V_i)} -\sum_{l \in \mathcal{N}_i}B_{il}I_{l}\\
L_{ti}\dfrac{dI_{ti}}{dt} &= -V_{i}-{R_{ti}}I_{ti}+V_{ti}\\
\end{aligned}
\right. ,
\end{equation}
\end{small}
where  $V_{ti}$ is the command to the Buck converter and $I_{ti}$ is the filter (generator) current. The terms $R_{ti} \in \mathbb{R}_{>0}$, $L_{ti} \in \mathbb{R}_{>0}$, and $C_{ti} \in \mathbb{R}_{>0}$ are the internal resistance, capacitance (lumped with the line capacitances), and inductance of the DGU converter. Each of these DGUs is equipped with local voltage regulators, which forms the \textit{primary control layer}. The main objective these controllers is to ensure that the voltage at each DGU's PC tracks a reference voltage $V_{ref,i}$ (modified by the \textit{secondary controller}; see Section \ref{sec:secondary} for more details).  For this purpose, as in \cite{Nahata, Tucci2016independent}, we augment each DGU with a multivariable PI regulator
	\begin{subequations}
	\begin{align}
	\dot{v}_i = \subss{e}{i} &= {V_{ref,i}}-{V}_i\label{eq:intdynamics},\\
	\subss{\CC}{i}:~V_{ti} &=K_{[i]}\subss{\hat{x}}{i},\label{eq:controldec},
	\end{align}
	\end{subequations}
where $\subss{\hat{x}}{i}=\left[V_i\text{ }I_{ti}\text{ }v_i\right]^T\in\Rset^{3}$ is the state of augmented DGU and $K_{[i]}=\left[k_{1,i}\text{ }k_{2,i}\text{ }k_{3,i}\right]\in\Rset^{1\times3}$ is the feedback gain. From \eqref{eq:DGUdynamics}-\eqref{eq:controldec}, the closed-loop DGU model is obtained as 
\begin{small}
\begin{equation}
\label{eq:DGUdynamicsupdated}
\begin{split}
\subss{\hat{\Sigma}}{i}^{DGU}:
&\left\{
\begin{aligned}
\dfrac{dV_{i}}{dt} &= \dfrac{1}{C_{ti}}I_{ti}-\dfrac{1}{C_{ti}}{I_{Li}(V_i)}- \dfrac{1}{C_{ti}}{\sum_{l \in \mathcal{N}_i}B_{il}I_{l}}\\
\dfrac{dI_{ti}}{dt} &= {\alpha_i}V_{i} + {\beta_i}I_{ti}+ {\gamma_i}v_{   i}\\
\dfrac{dv_i}{dt} &= -{V}_i+{V_{ref,i}}
\end{aligned}\right.\\
\end{split},
\end{equation}
\end{small}
where
\begin{small}
\begin{equation}
\label{eq:abg}
\alpha_i=\frac{(k_{1,i}-1)}{L_{ti}},~\beta_i=\frac{(k_{2,i}-R_{ti})}{L_{ti}},~\gamma_i=\frac{k_{3,i}}{L_{ti}}.
\end{equation} 
\end{small}
We note that the control architecture is decentralized since the computation of $V_{ti}$ uniquely requires the state of $\subss{\hat{\Sigma}}{i}^{DGU}$. 

\textit{Load model:}  In this work, we consider the standard ZIP model \cite{Nahata}. The parallel combination of these three loads is given as
\begin{equation}
\label{eq:loaddynamics}
I_{Li}(V_i)=\underbrace{Y_{Li}V_i}_{Z}+\underbrace{\bar{I}_{Li}}_{I} +\underbrace{V_i^{-1}P^*_{Li}}_{P}.
\end{equation}%
\begin{assum}
	\label{ass:voltage}
	The reference signals $V_{ref,i}$ and PC voltages $V_i$ are strictly positive for all $t \geq 0$.
\end{assum}
We remark that Assumption \ref{ass:voltage} is not a limitation, and rather reflects a common constraint in microgrid operation.  Indeed, negative PCC voltages reverse the role of loads and make them power generators. 
\section{Secondary control in DCmGs}
\label{sec:secondary}
\subsection{Problem formulation}
The primary control layer is designed to track a suitable reference voltage $V_{ref,i}$ at the $PC_i$. As such, they do not ensure current sharing and voltage balancing, defined as follows. 
\begin{defn} \textbf{(Current sharing \cite{Tucci2018, Zhao}).}
\label{defn:cs}
The load is said to be shared proportionally among DGUs if
\begin{equation}
\label{eq:cs_defn}
\frac{I_{ti}}{I_{ti}^s} = \frac{I_{tj}}{I_{tj}^s}\hspace{4mm}\text{for all }i,j\in \mathcal{V},
\end{equation}
where $I_{ti}^s>0$ is the rated current of $DGU_i$.
\end{defn}
Currents sharing ensures proportional sharing of loads amongst multiple DGUs.  This avoids situations of DGU overloading and prevents harm to the converter modules. As shown in the subsequent sections, in order to attain current sharing, the steady state voltages are not equal to $V_{ref,i}$. It is, however, desirable that PC voltages remain close to the nominal reference voltages for normal operation of the DCmG. To this aim, we state the objective of weighted voltage balancing in the following definition.
\begin{defn}\textbf{(Weighted voltage balancing \cite{Cucuzzella2018robust}).}
\label{defn:vb} The voltages are said to be balanced in the steady state if 
\begin{equation}
\label{eq:vb_defn}
\langle[I^s_{t}]{V}\rangle =\langle[I^s_{t}]V_{ref}\rangle .
\end{equation}
\end{defn}
\vspace{.2cm}
Voltage balancing implies that the weighted sum of PC voltages is equal to the the weighted sum of voltage references, ensuring boundedness of DCmG voltages. As noticed in \cite{Zhao}, in its absence, the PC voltages may experience drifts and increase monotonically despite the filter currents' being shared proportionally.
\begin{figure}[!h]
	\centering
	\ctikzset{bipoles/length=0.6cm}
	\tikzstyle{every node}=[font=\tiny]
	\vspace{0.1cm}
	\begin{circuitikz}[scale=0.48]
		\draw (1,1)  to [battery, o-o](1,4)
		to [short](1.5,4)
		to [short](1.5,4.5)
		to [short](3.5,4.5)
		to [short](3.5,0.5)
		to [short](1.5,0.5)
		to [short](1.5,4)
		to [short](1.5,1)
		to [short](1,1);
		\node at (2.5,2.5){ \footnotesize \textbf{Buck $i$}};
		\draw[-latex] (4,1.25) -- (4,3.75)node[midway,right]{$V_{ti}$};
		\draw (3.5,4) to [short](4,4)
		to [short](4.5,4)
		to [R=$R_{ti}$] (6,4)
		to [L=$L_{ti}$] (7.5,4)
		to [short, i=$\textcolor{green}{I _{ti}}$, -] (8.5,4)
		to [short](9,4) 
		to [C, l=$C_{ti}$, -] (9,1)
		to [short](4,1)
		to [short](3.5,1);
		\draw (12.3,4)  to [R=$R_{l}$] (14.5,4) to [short, i=${I _{l}}$](15,4)
		to [L=$L_{l}$] (16.5,4)
		to [short, -o] (17,4) node[anchor=north,above]{$V_j$};
		\draw (8.5,4) to (11,4) 
		to [ I ] (11 ,1)
		to [short] (9,1)
		to [short, -o] (17,1); 
		\draw (11,4) to [short](11.5,4);
		\draw (11,4) node[anchor=north, above]{$\textcolor{red}{V_i}$}  to [short, i_=$I_{Li}(V_i)$](11,2.9);
		\draw (11,4) to [short, i_=$I^*_{i}$](12.5,4);
		\node at (11,4.6)[anchor=north, above]{$PC_i$} ;
		\draw (11,4) to (12.3,4); 
		\draw[black, dashed] (.5,.25) -- (12.45,.25) -- (12.45,5.5) -- (.5,5.5)node[sloped, midway, above]{\footnotesize{ \textbf{DGU and Load $i$ }}}  -- (.5,.25);
		\draw[black, dashed] (12.7,.25) -- (16.5,.25) -- (16.5,5.5) -- (12.7,5.5)node[sloped, midway, above]{\footnotesize{ \textbf{Power line $l$}}}  -- (12.7,.25);
		\draw[red,o-] (10.9,4.15) -- (11.7,3) to (11.7,-1.5);
		\draw[red,latex-](8,-1.5)-- (9,-1.5) --  (10,-0.5)-- (11.7,-0.5);
		\draw[green,o-latex] (8.5,4.15) to (8.5,1.75) -- (8.5,-1) to (8,-1);
		\draw (10,-2) node(a) [black, draw,fill=white!20] {$\normalsize{\int}$};
		\draw[-latex] (a.west) to (8,-2);
		\draw (11.7,-2) node(b)[ circle, draw=black, minimum size=12pt, fill=lightgray!20]{};
		\draw[red, -latex] (11.7,1.75)  -- (b.north) node[pos=0.9, left]{\textcolor{black}{\normalsize{-}}};
		\draw[latex-] (a.east) -- (b.west);
		\draw[-latex] (12.8,-2)  -- (b.east) node[pos=0.7,above]{{+}} node[pos=0.25,right]{$V_{ref,i}$};
		\draw[fill=lightgray] (8,-2.25) -- (8,-0.75) -- (6.5,-1.5) -- (8,-2.25);
		\node at (7.5,-1.5) {$K_i$};
		\draw[fill=lightgray, draw=blue] (8,-4.25) -- (8,-2.75) -- (7,-3.5) -- (8,-4.25);
		\node at (7.7,-3.5) {\textcolor{blue}{$k_{4,i}$}};
		
		\draw (5.5,-1.5) node(c)[ circle, draw=black, minimum size=12pt, fill=lightgray!20]{};
		\draw[-latex] (6.5,-1.5) -- (c) node[pos=0.7,above]{{+}};
		\draw[-latex] (c.north)-- (5.5,1.5) -- (5.5,2.5) -- (5,2.5);
		\draw[-latex, blue] (7,-3.5) -- (5.5,-3.5) --(c)node[pos=0.8,left]{\textcolor{blue}{\normalsize{-}}};
		\draw[black, thick] (-0.5,-4.5) -- (17.5,-4.5) -- (17.5,6.5) -- (-0.5,6.5)node[sloped, midway, above]{\footnotesize{ \textbf{Primary Control}}}  -- (-0.5,-4.5);			
		
		\draw[black, thick] (-0.5,-5) -- (17.5,-5) -- (17.5,-10.5) --node[sloped, midway, below]{\footnotesize{ \textbf{Secondary Control}}} (-0.5,-10.5)  -- (-0.5,-5);

		\draw[blue, -latex] (11.5,-7)  -- (14.5,-7)   -- node[sloped, midway,above]{$\omega_i$} (14.5,-3.5) -- (11.7,-3.5) -- (b.south) node[pos=0.8,left]{\textcolor{black}{\normalsize{-}}};	
		\draw[blue, -latex] (11.7,-3.5)-- (8,-3.5);	
		
		\node[cloud, cloud puffs=16.2, cloud ignores aspect, minimum width=3cm, minimum height=1cm, opacity=0.75, align=center, fill=gray!20, thick, draw=black,text=black] (Gcloud) at (8.5, -7) {$\dot {\Omega} =\mathcal{L}_c[I^s_{t}]^{-1}I_t$ \\ $\omega=[I^s_{t}]^{-1}\mathcal{L}_c\Omega$};
		
		\tikzstyle{DGU} = [circle, draw, double, fill=red!20, draw=red]
		
		\draw (1,-6) node(D1) [DGU]  {\tiny $i$};
		\draw (1.25,-9) node(D3) [DGU]  {\tiny $j$};
		\draw (7,-9.25) node(D2) [DGU]  {\tiny $k$};
		\draw (16.5,-9.75) node(D4) [DGU]  {\tiny $l$};
		\draw (11.5,-9.8) node(D5) [DGU]  {\tiny $m$};
		
	\draw[darkgray,very thick, dotted] (D1) to  (Gcloud);
	\draw[darkgray,very thick, dotted] (D2) to  (Gcloud);
	\draw[darkgray,very thick, dotted] (D3) to  (Gcloud);
	\draw[darkgray,very thick, dotted] (D4) to  (Gcloud);
	\draw[darkgray,very thick, dotted] (D5) to  (Gcloud);
	\end{circuitikz}
	\caption{Schematic diagram showing primary and secondary control layers of the DCmG. Note that the topology of the communication network is not shown. }
	\label{fig:ctrl_complete}
	\vspace{-.3cm}
\end{figure}
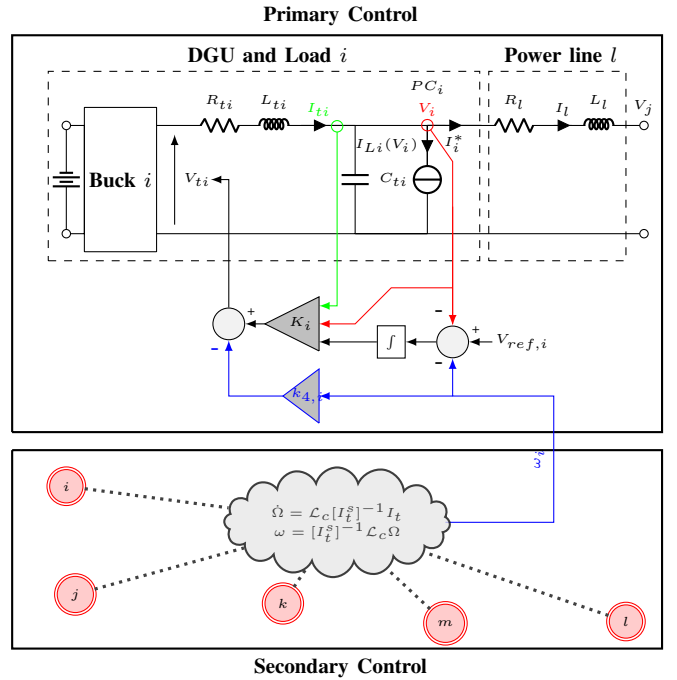	
\subsection{Consensus-based secondary control}
In order to achieve the previously stated objectives, we use a consensus-based secondary control layer. Consensus filters are commonly employed for achieving
global information sharing or coordination through distributed computations \cite{Bullo}. In our case, we propose the following 
consensus scheme 
\begin{equation}
\label{eq:basic_consensus}
\dot {\Omega}_i =\sum\limits_{j=1, j\neq i} ^{N}a_{ij}\left(\frac{I_{ti}}{I_{ti}^s}-\frac{I_{tj}}{I_{tj}^s}\right),
\end{equation}
where $a_{ij}>0$ if DGUs $i$ and $j$ are connected by a communication link ($a_{ij}=0$, otherwise). The corresponding \textit{communication graph} (see Figure \ref{fig:powernework}), assumed to be undirected and connected, is $\mathcal{G}_c=(\mathcal{D}, \mathcal{E}_c, W_c)$ where $(i,j)\in\mathcal{E}_{c}\Longleftrightarrow a_{ij}> 0$ and $W_c =\text{diag}\{a_{ij}\}$. Note that the topologies of $\mathcal{G}_c$ and $\mathcal{G}_{e}$ can be completely different. As shown in Figure \ref{fig:ctrl_complete}, the consensus variable 
\begin{equation}
\label{eq:omega}
\omega_i=\frac{1}{I_{ti}^s}\sum\limits_{j=1, j\neq i} ^{N}a_{ij}\left({\Omega_i}-{\Omega_j}\right)
\end{equation}
modifies the primary voltage controllers \eqref{eq:intdynamics} and \eqref{eq:controldec} as follows
\begin{subequations}
	\begin{align}
	\dot{v}_i &= {V_{ref,i}}-{V}_i-\omega_i\label{eq:intdynamicsconsensus}\\
V_{ti}(t)&=K_{[i]}\subss{\hat{x}}{i}-k_{4,i}\omega_i,\label{eq:controldecconsensus}
	\end{align}
\end{subequations}
where $k_{4,i} \in \mathbb{R}$. Consequently, using equations \eqref{eq:intdynamicsconsensus} and \eqref{eq:controldecconsensus}, one obtains the modified DGU dynamics as 
\begin{small}
\begin{equation}
\label{eq:DGUconsensus}
\begin{split}
\subss{\hat{\Sigma}}{i}^{DGU}:
&\left\{
\begin{aligned}
\dfrac{dV_{i}}{dt} &= \dfrac{1}{C_{ti}}I_{ti}-\dfrac{1}{C_{ti}}{I_{Li}(V_i)}- \dfrac{1}{C_{ti}}{\sum_{l \in \mathcal{N}_i}B_{il}I_{l}}\\
\dfrac{dI_{ti}}{dt} &= {\alpha_i}V_{i} + {\beta_i}I_{ti}+ {\gamma_i}v_{   i}+\delta_i\omega_i\\
\dfrac{dv_i}{dt} &= -{V}_i+{V_{ref,i}}-\omega_i
\end{aligned}\right.\\
\end{split},
\end{equation}
\end{small}
with
\begin{equation}
\label{eq:delta}
\delta_i=\dfrac{k_{4,i}}{L_{ti}}.
\end{equation}%
The complete dynamics of the DCmG under primary and secondary control are given by \eqref{eq:powerline} along with \eqref{eq:basic_consensus}-\eqref{eq:delta}. These equations can compactly be rewritten as
\begin{equation}
\label{eq:globalstatespace}
\dot{X}=\mathcal{A}X+\BB(V),
\end{equation}
where $X= \begin{bmatrix}
{V}^T
&{I_t}^T
&{v}^T
&{I}^T
&{\Omega}^T
\end{bmatrix}^T\in \mathbb{R}^{4N+M}, $
\begin{align*}
\begin{small}
\mathcal{A}=\underbrace{\begin{bmatrix}
	-C_t^{-1}Y_{L} &C_t^{-1} &\textbf{0} &-C_t^{-1}B &\textbf{0}\\
	[\alpha] &[\beta] &[\gamma] &\textbf{0} &[\delta][I^s_{t}]^{-1}\mathcal{L}_c\\
	-\textbf{I} &\textbf{0}  &\textbf{0}  &\textbf{0}  &-[I^s_{t}]^{-1}\mathcal{L}_c\\
	L^{-1}B^T &\textbf{0} &\textbf{0}  &-L^{-1}R &\textbf{0}\\
	\textbf{0} &\mathcal{L}_c[I^s_{t}]^{-1}&\textbf{0}&\textbf{0}&\textbf{0}
	\end{bmatrix}}_{\mathcal{A} \in \mathbb{R}^{(4N+M) \times (4N+M)}}
, \end{small}
\end{align*}
and
\begin{align*}
\begin{small}
\BB(V)=
\underbrace{\begin{bmatrix}
	-C_t^{-1}(\bar{I}_{L}+[V^{-1}]P^*_{L})\\
	\textbf{0}_N\\
	V_{ref}\\
	\textbf{0}_M\\
	\textbf{0}_N
	\end{bmatrix}}_{\mathcal{B}(V) \in \mathbb{R}^{(4N+M)}}.
\end{small}
\end{align*}
Note that $V \in \mathbb{R}^N$,  $V_{ref} \in \mathbb{R}^N$, $I_t \in \mathbb{R}^N$, $v \in \mathbb{R}^N$, $I \in \mathbb{R}^M$, ${P}^*_{L} \in\mathbb{R}^N$, $\bar{I}_{L} \in\mathbb{R}^N$, $\alpha \in \mathbb{R^N}$,  $\beta\in \mathbb{R^N}$, $\gamma \in \mathbb{R^N}$  are vectors of PC voltages, reference voltages, filter currents, integrator states, line currents, load powers, load currents, and parameters $\alpha_i$, $\beta_i$, $\gamma_i$, respectively. The matrices $R \in \mathbb{R}^{M \times M}_{>0}$, $L \in \mathbb{R}^{M \times M}_{>0}$, $Y_L \in \mathbb{R}^{N \times N}_{>0}$,  and $C_t \in \mathbb{R}^{N \times N}_{>0}$ are diagonal matrices collecting electrical parameters $R_l$, $L_l$, $Y_{Li}$, and $C_{ti}$, respectively. The matrix ${B} \in \mathbb{R}^{N \times M}$ is the incidence matrix of the electrical network and $\mathcal{L}_c \in \mathbb{R}^{N \times N}$  is the Laplacian matrix of the communication network.
\subsection{Analysis of Equilibria}
\label{sec:eq}
Prior to analyzing the stability of the closed-loop system \eqref{eq:globalstatespace}, it is pertinent to establish that an equilibrium exists such that both the objectives  \eqref{eq:cs_defn} and \eqref{eq:vb_defn} are jointly attained. We emphasize that, in a primary-controlled DCmG, a reference voltage $V_{ref,i}$ is directly enforced at the $i^{th}$ PC. Thus, a unique equilibrium point always exists; see \cite{Nahata}. On the contrary, once the secondary layer is activated, the voltage references are tweaked by $\omega_i$ (see \eqref{eq:intdynamicsconsensus}), governed by equations \eqref{eq:basic_consensus} and \eqref{eq:omega}. Since the presence of ZIP loads renders the DCmG dynamics nonlinear, it may occur that an equilibrium point fails to exist (see Section \ref{sec:simulations} for a simulation example). Hence, in this section, we pursue whether the closed-loop system \eqref{eq:globalstatespace} possesses an equilibrium point, and if so, under what conditions on loads, topology of electrical and communication networks, and controller gains. 
\begin{lem}
	\label{lem:equilibriumchar}
\textbf{(Steady-state behavior of the DCmG).}	Consider the DCmG dynamics \eqref{eq:globalstatespace}. The following statements hold: 
	\begin{enumerate}
		\item In steady state, the objectives \eqref{eq:cs_defn} and \eqref{eq:vb_defn} are attained;
		\item A steady-state solution $\bar{X}=[\bar{V}^T, \bar{I}_t^T, \bar{v}^T, \bar{I}^T, \bar{\Omega}^T]^T$ exists only if $[\gamma_i]$ is invertible and there exists a $\bar{V}$ concurrently satisfying the following equations
		\begin{subequations}
			\label{eq:equilibriumexistence}
			\begin{equation}
			\mathcal{L}_e\bar{V}+\mathcal{L}_{t}[I^s_{t}]^{-1}([\bar{V}^{-1}]P_L^*+\bar{I}_L+Y_L\bar{V}) =0, \label{eq:eqpowerflow}
			\end{equation}
			\begin{equation}
			\textbf{1}^T_N[I^s_{t}]\bar{V}=\textbf{1}^T_N[I^s_{t}]V_{ref},\label{eq:voltagebalancing}
			\end{equation}	
		\end{subequations}
		where  $\mathcal{L}_t=[I^s_{t}]-(\textbf{1}_N^T[I^s_{t}]\textbf{1}_N)^{-1}[I^s_{t}]\textbf{1}_N\textbf{1}_N^T[I^s_{t}]$, and $\mathcal{L}_e=BR^{-1}B^T$ is the Laplacian of the electric network. 
	\end{enumerate}
\end{lem}
\begin{proof} 
	Any steady state solution of \eqref{eq:globalstatespace} satisfies
	\begin{subequations}
		\begin{align}
		-Y_{L}\bar{V} 	-\bar{I}_{L}-[\bar{V}^{-1}]P^*_{L}+{\bar{I}_t} -B\bar{I} &=0\label{eq:eq1},\\
		[\alpha]\bar{V} +[\beta] \bar{I_t}+[\gamma]	\bar{v} +[\delta][I^s_{t}]^{-1}\mathcal{L}_c\bar{\Omega}&=0\label{eq:eq2},\\
		V_{ref}-\bar{V} -[I^s_{t}]^{-1}\mathcal{L}_c\bar{\Omega}&=0\label{eq:eq3},\\
		B^T \bar{V}-R\bar{I} &=0\label{eq:eq4},\\
		\mathcal{L}_c[I^s_{t}]^{-1}	\bar{I_t}&=0.\label{eq:eq5}
		\end{align}
	\end{subequations}
	One has from \eqref{eq:eq5} that $\bar{I}_t=\epsilon [I^s_t]\textbf{1}_N$ for some $\epsilon\in\mathbb{R}$, warranting the attainment of \eqref{eq:cs_defn}. Since $\textbf{1}^T_NB=\textbf{0}_M$, \eqref{eq:eq1} implies that $\textbf{1}_N^T\bar{I}_t=\textbf{1}_N^T(Y_{L}\bar{V} +\bar{I}_{L}+[\bar{V}^{-1}]P^*_{L})$, then $\epsilon=(\textbf{1}_N^T[I^s_{t}]\textbf{1}_N)^{-1}\textbf{1}_N^T(Y_{L}\bar{V} +\bar{I}_{L}+[\bar{V}^{-1}]P^*_{L}).$ We can equivalently represent 
	\begin{equation}
	\label{eq:Itbar}
	\bar{I}_t=(\textbf{1}_N^T[I^s_{t}]\textbf{1}_N)^{-1}[I^s_{t}]\textbf{1}_N\textbf{1}_N^T(Y_{L}\bar{V} +\bar{I}_{L}+[\bar{V}^{-1}]P^*_{L}).
	\end{equation}
	Using \eqref{eq:eq4},
	\begin{equation}
	\label{eq:Ibar}
	\bar{I}=R^{-1}B^T\bar{V}.
	\end{equation}
	On substituting \eqref{eq:Itbar} and \eqref{eq:Ibar} into \eqref{eq:eq1}, one obtains \eqref{eq:eqpowerflow}. Moreover, for an $\bar{\Omega}$ to exist such that \eqref{eq:eq3} holds, $[I^s_t](V_{ref}-\bar{V})\in H^1$, which yields \eqref{eq:voltagebalancing} and guarantees \eqref{eq:vb_defn} in steady state.  If there exists a $\bar{V}$ solving \eqref{eq:equilibriumexistence}, $\bar{I}_t$ and $\bar{I}$ exist due to \eqref{eq:Itbar} and \eqref{eq:Ibar}, respectively. As \eqref{eq:voltagebalancing} holds, from \eqref{eq:eq3}, an equilibrium vector  $\bar{\Omega}=\mathcal{L}_c^\dagger[I^s_t](V_{ref}-\bar{V})+\eta\textbf{1}_N, \eta \in \mathbb{R}$ exists. Finally, on substituting $\bar{V}, \bar{I}_t, \bar{I},$ and $\bar{\Omega}$ into \eqref{eq:eq2}, one has $\bar{v}=[\gamma]^{-1}\left(([\alpha]+[\delta])\bar{V}-[\delta]V_{ref} +[\beta] \bar{I}_t\right)$.
\end{proof}
\begin{rmk}\textbf{(Solvability of \eqref{eq:equilibriumexistence}).}
	We remark that \eqref{eq:eqpowerflow} represents the DC power-flow equations when DGU currents are shared proportionally. The existence and uniqueness of solutions of the power-flow equations have been tackled in \cite{simpson2016voltage, LaBella}. As shown in what follows, the tools therein cannot be applied directly to ascertain the solvability of \eqref{eq:equilibriumexistence} as \eqref{eq:voltagebalancing} restricts the voltage solutions on a hyperplane.
\end{rmk}
To analyze the existence of a voltage solution to \eqref{eq:equilibriumexistence}, we rewrite it as
\begin{equation}
\label{eq:PFmatrix}
\tilde{\mathcal{L}}V=\tilde{I}-\tilde{\LL}_t[V^{-1}]P_L^*,
\end{equation}
where $\tilde{\mathcal{L}}=\begin{bmatrix}
\LL_p\\
\textbf{1}^T_N[I^s_{t}]
\end{bmatrix}$, $\LL_p=\mathcal{L}_e+\mathcal{L}_{t}[I^s_{t}]^{-1}Y_L$, $\tilde{I}=\begin{bmatrix}-\mathcal{L}_{t}[I^s_{t}]^{-1}\bar{I}_L\\
\textbf{1}^T_N[I^s_{t}]V_{ref}
\end{bmatrix}$, and $\tilde{\LL}_t=\begin{bmatrix}\mathcal{L}_{t}[I^s_{t}]^{-1}\\
\textbf{0}\end{bmatrix}$. 
\begin{thm}
	\label{thm:existence}
	\textbf{(Existence and uniqueness of a voltage solution).} Consider \eqref{eq:PFmatrix} along with the vector $V^*=\tilde{\LL}^\dagger\tilde{I}$. Assume that $[V^*]$ is invertible and define $P_{cri}=4[V^*]^{-1}\tilde{\LL}^\dagger\tilde{\LL}_t[V^*]^{-1}$. Assume that the network parameters and loads satisfy 
	\begin{equation}
	\label{eq:Criticalpower}
	\Delta=||P_{cri}P_L^*||_{\infty}<1,
	\end{equation}
	and define the percentage deviations $\delta_{-}\in[0,\frac{1}{2})$ and $\delta_{+}\in(\frac{1}{2},1]$ as the unique solutions of $\Delta=4\delta_{\pm}(1-\delta_{\pm})$. The following statements hold:
	\begin{enumerate}
		\item[1)] There exists a unique voltage solution $V \in \HH(\delta_{-})$ of \eqref{eq:PFmatrix}, where
		\begin{equation}
		\label{eq:existenceset}
		\HH(\delta_{-})\coloneqq \{V\in \mathbb{R}^N |  (1-\delta_{-})V^*\leq V \leq(1+\delta_{-})V^*\}.
		\end{equation}
		Moreover, there exist no solutions of \eqref{eq:PFmatrix} in the open set 
		\begin{equation}
		\label{eq:nonexistenceset}
		\II\coloneqq \{V\in \mathbb{R}^N |  ( V >(1-\delta_{+})V^* ~\text{and}~V\notin \HH(\delta_{-}) \};
		\end{equation}
		\item[2)]  For $P_L^*=0$, $V^*$ is the unique solution of \eqref{eq:PFmatrix};
		\item [3)] If $(1-\delta_{+})V^* <V_{ref}$, then, there exist no solutions of \eqref{eq:PFmatrix} in the closed set 
		\begin{equation}
		\label{eq:nonexistenceset1}
		\JJ\coloneqq \{V\in \mathbb{R}^N |  ( V \leq (1-\delta_{+})V^* \}.
		\end{equation} 
	\end{enumerate}
	\end{thm}
\begin{proof}
	 For the sake of brevity, the proof is omitted, and can be found in \cite{NahataConsensusTech}.
\end{proof}
\begin{rmk}
	Under the sufficient conditions provided in Theorem \ref{thm:existence}, the existence of an equilibrium point depends upon the critical power matrix $P_{cri}$ and the power absorption $P_L^*$. Clearly, from \eqref{eq:PFmatrix}, the communication network topology $\GG_c$ has no impact on $P_{cri}$. 
\end{rmk}
 Hereafter, in order to be in line with Assumption \ref{ass:voltage}, we assume that $V^*$ is positive for $V_{ref}>0$. 
\subsection{Stability of the DCmG network}
In this section, we aim to study the stability of the closed-loop system \eqref{eq:globalstatespace}, necessary for the DCmG to exhibit the desired steady-state behavior described in Section \ref{sec:eq}. 
\begin{thm}
	\label{thm:stability}
	\textbf{(Stability of the closed-loop DCmG).} Consider the closed-loop system \eqref{eq:globalstatespace} resulting from equations \eqref{eq:basic_consensus}-\eqref{eq:DGUconsensus}, along with Assumption \ref{ass:voltage}. For $i \in \DD$, if the feedback gains $k_{1,i},~k_{2,i}$, and $k_{3,i}$ belong to the set 
	\begin{equation}
	\label{eq:explicitgains}
	\ZZ_{[i]}= \left\{ \begin{array}{l}
	k_{1,i}<1,\\
	k_{2,i}<R_{ti},\\
	0<k_{3,i}<\frac{1}{L_{ti}}(k_{1,i}-1)(k_{2,i}-R_{ti})
	\end{array} \right\},
	\end{equation} $k_{4,i}=k_{1,i}$ and the Z and P components of  \eqref{eq:loaddynamics} verify
	\begin{equation}
	\label{eq:maximumpowerconsumption}
	P^*_{Li}< Y_{Li}\bar{V}_{i}^2,
	\end{equation} then the equilibrium point $\bar{X}$ is locally asymptotically stable, and is globally asymptotically stable when $\bar{P}^*_L=0$.
\end{thm}
\begin{proof}
	The proof is provided in \cite{NahataConsensusTech} and is skipped due to space constraints.
\end{proof}
\begin{rmk}\textbf{(Compatibility with primary control and behavior under communication collapse).} Equations \eqref{eq:powerline} and \eqref{eq:DGUdynamicsupdated} represent the DCmG under primary control when the secondary layer is inactive. As shown in \cite{Nahata}, \eqref{eq:explicitgains} and \eqref{eq:maximumpowerconsumption} also make it possible to design stabilizing primary controllers. This enables us to reach the following conclusions: (i) the design of the proposed secondary controllers is fully compatible with the primary layer, and solely requires setting an additional control gain $k_{4,i}=k_{1,i}\in \ZZ_{[i]}$ once activated; (ii) if the DCmG undergoes a communication collapse, the primary controllers maintain voltage stability without any human intervention and force each PC to track $V_{ref,i}$ in steady state. 
\end{rmk}

\section{Simulation results}
\label{sec:simulations}

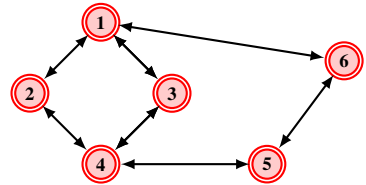
\begin{figure}[h!]
	\centering
	\begin{tikzpicture}[scale=0.7,transform shape,->,>=stealth',shorten >=1pt,auto,node distance=1.25cm, thick,main node/.style={circle, draw, double, fill=red!20, draw=red}]
	
	\node[main node] (1) {\textbf{1}};
	\node[main node] (2) [below left= of 1] {\textbf{2}};
	\node[main node] (3) [below right= of 1] {\textbf{3}};
	\node[main node] (4) [below right= of 2] {\textbf{4}};
	\node[main node] (5) [right=2.5 of 4] {\textbf{5}};
	\node[main node] (6) [above right=1.5 and of 5] {\textbf{6}};
	
	\path[latex-latex,every node/.style={font=\sffamily\small}]
	(1) edge node [left] {} (3)
	(1) edge node [left] {} (2)
	(3) edge node [left] {} (4)
	(2) edge node [left] {} (4)
	(1) edge node [left] {} (3)
	(3) edge node [left] {} (4)
	(5) edge node [right] {} (4);
	
	\path [latex-latex, thick] 
	(1) edge (6)
	(6) edge (5);
	\vspace{-2cm}
	\end{tikzpicture}
	\caption{Simplified DCmG composed of 6 DGUs.}
	\label{fig:5areasplug}
\end{figure}
In this section, we aim to demonstrate the capability of the proposed control scheme to guarantee current sharing and voltage balancing. We consider a meshed DCmG composed of 6 DGUs (see Figure \ref{fig:5areasplug}) with non-identical electrical parameters, adopted from \cite{Tucci2016independent}. In our simulations, we assume ZIP loads with powers $P^{*}_{Li}$, $i = 1, \dots, 6$ always fulfilling  \eqref{eq:maximumpowerconsumption}. The PC voltage references $V_{ref,i}$, $i = 1, \dots, 6$ are set close to 50V and to slightly different values. We highlight that, the control gains $k_{1,i},~k_{2,i}$, and $k_{3,i}$ belong to the set $\mathcal{Z}_{[i]}$ defined in \eqref{eq:explicitgains}. For the considered network, the variable $\Delta$ is much smaller than 1, guaranteeing the existence of a voltage solution to \eqref{eq:PFmatrix}.  In the following discussion, we evaluate voltage balancing and current sharing in the DCmG when DGUs are plugged-in and loads are arbitrarily changed.

\textit{Plug-in of all DGUs:} At time $t<0$, all the DGUs are isolated and only the primary voltage regulators, designed as in \cite{Nahata}, are active tracking $V_{ref,i}$ at their respective PCCs. At time $t=0$, all the DGUs are connected to from the DCmG shown in Figure \ref{fig:5areasplug} and the secondary control layer is activated. The control gain $k_{4,i}$ is set equal to $k_{1,i}$ for all DGUs, and no other control gain is modified. As shown in Figure \ref{fig:consensuscurrents}, the DGU currents are shared proportional to their capacity. This is achieved by automatically adjusting the reference voltages at PC (see Fig \ref{fig:consensusvoltages}). Note that the voltages $V^{max}$ and $V^{min}$ represent the maximum and minimum elements of $(1-\delta_{-})V^*$ and $(1+\delta_{-})V^*$, respectively. Although not shown, each PC voltage is bounded as in \eqref{eq:existenceset}. Moreover, Figure \ref{fig:consensusregulations} shows that weighted voltage balancing is achieved in steady state.

\textit{Robustness to changes in load:} At $t=2$s, an increase takes place in load consumption. Indeed, as in Figure \ref{fig:secctrl}, the filter currents attain a new steady state and share the total load proportionally. In addition, the PC voltages are bounded and converge to a new steady state. We highlight that, since the voltage references are left unchanged, the weighted voltage sum remains the same.
\begin{figure}[h]
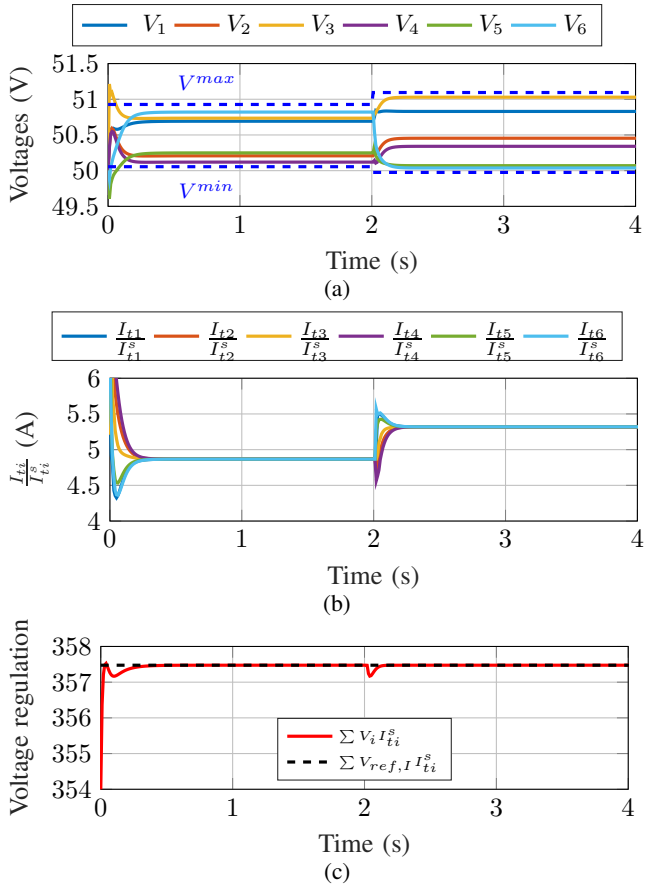

	\definecolor{mycolor1}{rgb}{0.00000,0.44700,0.74100}%
	\definecolor{mycolor2}{rgb}{0.85000,0.32500,0.09800}%
	\definecolor{mycolor3}{rgb}{0.92900,0.69400,0.12500}%
	\definecolor{mycolor4}{rgb}{0.49400,0.18400,0.55600}%
	\definecolor{mycolor5}{rgb}{0.46600,0.67400,0.18800}%
	\definecolor{mycolor6}{rgb}{0.30100,0.74500,0.93300}%
	\setlength\fheight{1.9cm} 
	\setlength\fwidth{0.4\textwidth}
	\centering
	\ref{voltages}\\
	\begin{subfigure}[b]{0.3\textwidth}
		\hspace{-1.8cm}\input{voltage.tex}
		\vspace{-0.7cm}
		\caption{}
		\label{fig:consensusvoltages}
		\vspace{0.1cm}
	\end{subfigure}
	\ref{currents}\\
		\setlength\fheight{1.9cm} 
		\setlength\fwidth{0.4\textwidth}
	\begin{subfigure}[b]{0.3\textwidth}
	\hspace{-1.8cm}\input{currentsharing.tex}
	\vspace{-0.7cm}
	\caption{}
	\label{fig:consensuscurrents}
	\end{subfigure}
	~ 
	\begin{subfigure}[b]{0.3\textwidth}
		\vspace{-0.1cm}
		\hspace{-1.8cm}
%
%
\begin{tikzpicture}

\begin{axis}[%
width=0.986\fwidth,
height=\fheight,
at={(0\fwidth,0\fheight)},
scale only axis,
xmin=0,
xmax=4,
xtick={0, 1, 2, 3, 4},
xlabel style={font=\color{white!15!black}},
xlabel={Time (s)},
ymin=354,
ymax=358,
ylabel style={font=\color{white!15!black}},
ylabel={Voltage regulation},
axis background/.style={fill=white},
xmajorgrids,
ymajorgrids,
legend style={legend cell align=left, align=left, draw=white!15!black, at={(0.5,0.5)}, anchor=north}
]
\addplot [color=red, line width=1.2pt]
  table[row sep=crcr]{%
0	354.000000003047\\
0.00999999999999091	356.36910782681\\
0.0199999999999818	357.212518504969\\
0.0299999999999727	357.485928645298\\
0.0400000000000205	357.52466567456\\
0.0500000000000114	357.438174177689\\
0.0600000000000023	357.342024639896\\
0.0699999999999932	357.257909551977\\
0.0799999999999841	357.20273291357\\
0.089999999999975	357.173417109106\\
0.100000000000023	357.165493787811\\
0.110000000000014	357.172366509928\\
0.120000000000005	357.188622363512\\
0.129999999999995	357.209950719084\\
0.160000000000025	357.279411748988\\
0.180000000000007	357.319273730644\\
0.189999999999998	357.336341929906\\
0.199999999999989	357.351577129928\\
0.20999999999998	357.365125010995\\
0.220000000000027	357.377152479767\\
0.240000000000009	357.397303127664\\
0.259999999999991	357.41321154833\\
0.279999999999973	357.425814951708\\
0.300000000000011	357.435829877965\\
0.319999999999993	357.443800982118\\
0.339999999999975	357.450149514456\\
0.370000000000005	357.457334755134\\
0.399999999999977	357.462442344673\\
0.439999999999998	357.467034873134\\
0.490000000000009	357.470490828762\\
0.550000000000011	357.472722033652\\
0.639999999999986	357.474182845075\\
0.810000000000002	357.474879185358\\
1.55000000000001	357.474999243215\\
2.00999999999999	357.474999583035\\
2.01999999999998	357.471952111853\\
2.02999999999997	357.238697059099\\
2.04000000000002	357.170084842488\\
2.05000000000001	357.183561834401\\
2.06	357.231298998845\\
2.06999999999999	357.286957818351\\
2.07999999999998	357.337113766586\\
2.08999999999997	357.378186040195\\
2.10000000000002	357.409161574217\\
2.11000000000001	357.431488284423\\
2.12	357.44686280067\\
2.13	357.457109392231\\
2.13999999999999	357.463718090151\\
2.14999999999998	357.467862074902\\
2.16000000000003	357.470398370614\\
2.18000000000001	357.472812213148\\
2.20999999999998	357.473860357587\\
2.32999999999998	357.474701381611\\
2.69	357.47499435637\\
4.00999999999999	357.474999376299\\
};
\addlegendentry{\tiny $\sum{V_iI^{s}_{ti}}$}

\addplot [color=black, dashed, line width=1.2pt]
  table[row sep=crcr]{%
0	357.475\\
4.00999999999999	357.475\\
};
\addlegendentry{\tiny $\sum{V_{ref,I}I^{s}_{ti}}$}

\end{axis}

\begin{axis}[%
width=1.272\fwidth,
height=1.272\fheight,
at={(-0.165\fwidth,-0.177\fheight)},
scale only axis,
xmin=0,
xmax=1,
ymin=0,
ymax=1,
axis line style={draw=none},
ticks=none,
axis x line*=bottom,
axis y line*=left,
legend style={legend cell align=left, align=left, draw=white!15!black}
]
\end{axis}
\end{tikzpicture}%
		\vspace{-0.7cm}
		\caption{}
		\label{fig:consensusregulations}
	\end{subfigure}
	\caption{PCC voltages, weighted filter currents, and weighted voltage sum under secondary control. \vspace{-0.2cm}}\label{fig:secctrl}
\end{figure}
With the intention of demonstrating that a steady-state voltage may not exist under secondary control due to the presence of P loads, we increase the line resistances multi-fold to violate \eqref{eq:Criticalpower}. In such a scenario, the PC voltages do not reach a steady state; see Figure \ref{fig:voltagecollapse}. In particular, the voltage $V_6$ approaches 0, leading to an increasing power absorption by the P load at PC 6, and eventually to a simulation failure. 
\begin{figure}[!h]
	\vspace{-0.1cm}
	\definecolor{mycolor1}{rgb}{0.00000,0.44700,0.74100}%
	\definecolor{mycolor2}{rgb}{0.85000,0.32500,0.09800}%
	\definecolor{mycolor3}{rgb}{0.92900,0.69400,0.12500}%
	\definecolor{mycolor4}{rgb}{0.49400,0.18400,0.55600}%
	\definecolor{mycolor5}{rgb}{0.46600,0.67400,0.18800}%
	\definecolor{mycolor6}{rgb}{0.30100,0.74500,0.93300}%
	\setlength\fheight{1.9cm} 
	\setlength\fwidth{0.4\textwidth}
	\centering
	{\hspace{-.3cm}\input{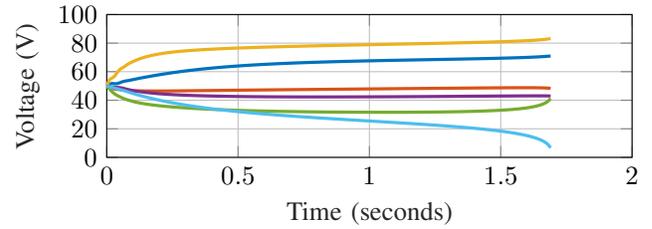}}\\
	\caption{Voltage collapse in the DCmG network due to nonexistence of a steady state.
	\vspace{-0.5cm}}
	\label{fig:voltagecollapse}
\end{figure}

\section{CONCLUSIONS}
\label{sec:conclusions}
In this paper, a novel secondary consensus-based control layer for current sharing and voltage balancing in DCmGs was presented. We considered a DCmG composed of realistic DGUs, RLC lines, and ZIP loads. A rigorous steady-state analysis was conducted, and appropriate conditions were derived to ensure the attainment of both objectives. In addition, a voltage stability analysis was provided showing that the controllers can be synthesized in a decentralized fashion.
Future developments will study the impact of non-idealities (such as transmission delays, data quantization and packet drops) on the performance of closed-loop mGs. Further developments can also consider the inclusion of Boost and other DC-DC converters.
                                  
\bibliographystyle{IEEEtran}
\bibliography{articleref}
\end{document}